\documentclass[11pt]{article}
\usepackage{amsmath,amssymb,latexsym,amsfonts,color,natbib,dsfont,url}
\usepackage{graphicx,float}
\usepackage[table]{xcolor}
\usepackage{hyperref}
\DeclareMathOperator*{\argmin}{arg\,min}
\allowdisplaybreaks
\title{$k$-means clustering of extremes}
\author{
Anja Jan{\ss}en\footnote{KTH Royal Institute of Technology; Department of Mathematics, Lindstedtsv\"agen 25, 20 400 Stockholm, Sweden;
email: anjaj@kth.se}
\and
Phyllis Wan\footnote{Erasmus University Rotterdam; Erasmus School of Economics, Burg.\ Oudlaan 50, 3062 PA Rotterdam, the Netherlands;
email: wan@ese.eur.nl}}

\usepackage[thmmarks,amsthm, amsmath]{ntheorem}

\theorempreskipamount 2ex \theorempostskipamount 3ex \theoremsymbol{$\Box$}
\theoremseparator{\quad}

\theorembodyfont{\em}

\newtheorem{theorem}{Theorem}[section]

\newtheorem{remark}[theorem]{Remark}
\newtheorem{proposition}[theorem]{Proposition}

\numberwithin{equation}{section}
\theorembodyfont{\rm}

\def\Min(#1,#2){#1\wedge #2}
\def\Max(#1,#2){#1\vee #2}

\begin{document}

\maketitle

\begin{abstract}
The $k$-means clustering algorithm and its variant, the spherical $k$-means clustering, are among the most important and popular methods in unsupervised learning and pattern detection. In this paper, we explore how the spherical $k$-means algorithm can be applied in the analysis of only the extremal observations from a data set. By making use of multivariate extreme value analysis we show how it can be adopted to find ``prototypes'' of extremal dependence and derive a consistency result for our suggested estimator. In the special case of max-linear models we show furthermore that our procedure provides an alternative way of statistical inference for this class of models. Finally, we provide data examples which show that our method is able to find relevant patterns in extremal observations and allows us to classify extremal events. 
\end{abstract}

\section{Introduction}\label{Sec:intro}
When looking at multivariate and in particular high-dimensional data, it is one of the most important objectives  of a statistical analysis to detect structures and patterns in the observations so as to simplify their complexity. This task has led to the development of an abundance of procedures in the field of unsupervised learning, see for example \cite{Hastieetal09} for an overview. Usually the analysis is focused on describing the bulk of the data and one looks for results that apply to most observations at hand. However, when \textit{extremal} observations are of interest, a different approach needs to be taken.  In this paper we consider the complexity reduction of extremal observations via clustering. 

For a specific dataset, a naive implementation is to choose the observations with the largest norm (thereby considering them as extremal observations) and apply a clustering algorithm to them.  However, this proves to be inefficient as extremal points are typically spread out in space.
In the presence of heavy-tailed observations, most classical clustering algorithms would have further problems from the possibly infinite second moments. In order to allow for robust and meaningful estimation, one should incorporate structural results about the particular kind of data at hand into the estimation procedure.

Multivariate extreme value theory (MEVT) provides us with such a framework and has useful applications in a wide range of disciplines, such as finance and climate science, see for example \cite{Fougeres03} and \cite{Davisonetal12} for an overview. Most of the current parametric models and estimation methods focus on the bivariate or lower dimensional scenario and are difficult to generalize to higher dimensions due to either lack of flexibility or heavy computation loads, see \cite{DavisonHuser15}.  

Recently there have been a few attempts to adapt complexity reduction to extremal dependence. One way of doing so is applying classical dimension reduction techniques to (transformed) extremal observations, in the form of principal component analysis and related covariance matrix decomposition techniques, see \cite{Haugetal09} and \cite{CooleyThibaud16}, or empirical basis functions, see \cite{Morrisetal18}. Another direction of research is aimed at dividing the parameter space into lower dimensional subspaces via making use of the phenomenon of asymptotic independence (meaning that in extremal events there are often only a few components which are large at the same time), see \cite{Goixetal17} and \cite{Chiapinoetal19}.  \cite{Chautru15} identifies relevant subspaces by first reducing the dimension of projected observations by a spherical principal components procedure, then clustering the projected data using spherical $k$-means, and as a last step attributes a lower-dimensional subspace to each cluster. Finally, \cite{Bernardetal13} presents a classification approach with special emphasis on a spatial decomposition of separated regional clusters for extremal events. The methodology is based on a $k$-means like algorithm, where distances between two stations are measured by their $F$-madogram. 

The usage of $k$-means estimation in extremes is therefore not entirely new (see also \cite{Einmahletal12}, where the procedure is mentioned to produce starting values for numerical estimation algorithms), but so far it has only been applied as an intermediate step towards a specific goal and its theoretical properties have not been explored. The aim of this paper is twofold: First, we provide the theoretical background as to how a $k$-means algorithm applied to extremal observations can be constructed as a consistent estimator of theoretical extremal cluster centers, see Theorem \ref{Th:main}. Second, we demonstrate that these cluster centers themselves can be seen as prototypes of directions of extremal events and the algorithm therefore provides a comprehensive, computationally fast and robust procedure to interpret observed extremes. As a side effect we demonstrate that our procedure can be seen as an alternative, consistent way of estimating relevant components of max-linear models, which recently gained popularity in applications due to their relationship to causal models for extremes as implied by directed acyclic graph (DAG) models, see \cite{Gissibl18, Gissibletal18, GissiblKlueppelberg18}.

The paper is organized as follows: Section \ref{Sec:background} provides a short background on the two main components of our method, MEVT and the spherical $k$-means algorithm. In Section \ref{Sec:main}, we present a general consistency result for the spherical $k$-means algorithm in the extremal setting and construct a non-parametric estimator for the theoretical cluster centers. The application of our procedure to the particular class of max-linear models is outlined in Section \ref{Sec:max-linear}. Finally, we demonstrate the application and interpretation of our results by looking at three data examples in Section \ref{Sec:data}.

%%%%%%%%%%%%%%%%%%%%%%%%%%%%%%%%%%%%%%%%%%%%%%%

\section{Background}\label{Sec:background}
\subsection{Multivariate extreme value theory}
In order to describe the extremal behavior of a random vector, a typical assumption is that the componentwise maxima, generated from i.i.d.\ copies of this vector, converge jointly to a non-degenerate limit distribution after proper linear normalization. More formally, let $(X_1^i, \ldots, X_d^i), i \in \mathbb{N},$ be i.i.d.\ copies of the random vector $\mathbf{X}=(X_1, \ldots, X_d)$.  We assume that there exist sequences of constants $a_j^n>0, b_j^n \in \mathbb{R}, 1 \leq j \leq d, n \in \mathbb{N}$ and a (in each margin non-degenerate) distribution function $G$, such that 
\begin{equation}\label{Eq:maxlim} \lim_{n \to \infty}P\left(\frac{\max_{i=1,\ldots, n}X_1^i-b_1^n}{a_1^n} \leq x_1, \ldots, \frac{\max_{i=1,\ldots, n}X_d^i-b_d^n}{a_d^n}\leq x_d\right) = G(x_1, \ldots, x_d),
\end{equation}
for all continuity points $(x_1, \ldots, x_d)$ of $G$. 
We say that $\mathbf{X}$ is \textit{in the max-domain of attraction of the extreme value distribution} $G$. A central result of extreme value theory is that this convergence can be broken down into two separate components. First, all marginal distributions $G_i, 1 \leq i \leq d$, of $G$ have to be univariate extreme value distributions of the form
\begin{equation}
\label{Eq:univEVD}
G_i(x)=\exp\left(-\left(1+\gamma_i \frac{x-\mu_i}{\sigma_i}\right)^{-1/\gamma_i}\right), \;\;\; 1+\gamma_i \frac{x-\mu_i}{\sigma_i} >0, 
\end{equation}
with $\gamma_i, \mu_i \in \mathbb{R}, \sigma_i>0$, where for $\gamma_i=0$ the right hand side is interpreted as $\exp(-\exp(-(x-\mu_i)/\sigma_i)), x \in \mathbb{R}$. The parameter $\gamma_i$, known as the \textit{extreme value index}, is the most important parameter in describing the univariate extremal behavior of component $i$.
Here a wealth of statistical procedures exists for univariate extremes, see \cite{deHaanFerreira07}, Chapters 3 and 4, for an overview.  Second, let $F_i$ be the (continuous) marginal distributions of $X_i$, $i=1,\ldots,d$, then the convergence in \eqref{Eq:maxlim} holds if and only if the standardized vector 
\begin{equation}
\label{Eq:trans}
\mathbf{Y}=\left(\frac{1}{1-F_1(X_1)}, \ldots, \frac{1}{1-F_d(X_d)}\right) \end{equation} 
satisfies
\begin{equation}
\label{Eq:spectral}
\lim_{u \to \infty} P\left(\frac{\mathbf{Y}}{\|\mathbf{Y}\|} \in B \, \Big| \, \|\mathbf{Y}\| >u \right) = S (B), 
\end{equation} 
for a probability measure $S$ on $\mathbb{S}^{d-1}_+:=\{\mathbf{x} \in [0,\infty)^d: \|\mathbf{x}\|=1 \}$, where $\|\cdot\|$ stands for an arbitrary but fixed norm, and its continuity-Borel-sets $B$, see \cite{Beirlantetal06}, Chapter 8.  The measure $S$ in \eqref{Eq:spectral} thus describes the limiting behavior of the directions that we see in extremal observations from $\mathbf{Y}$. Furthermore, the measure $S$ is connected to the limiting behavior of maxima as described in the following.
The transformed $\mathbf{Y}$ satisfies
\begin{equation*} \lim_{n \to \infty}P\left(\frac{\max_{i=1,\ldots, n}Y_1^i}{n} \leq x_1, \ldots, \frac{\max_{i=1,\ldots, n}Y_d^i}{n}\leq x_d\right) = G_0(x_1, \ldots, x_d),
\end{equation*}
for i.i.d.\ copies $(Y_1^i, \ldots, Y_d^i), i \in \mathbb{N},$ of $\mathbf{Y}$, where $G_0$ is an extreme-value distribution with standard Fr\'{e}chet margins (i.e., $\gamma_i=\sigma_i=1, \mu_i=0$ in \eqref{Eq:univEVD}). For $G_0$ there exists a so-called \textit{exponent measure} $\nu$ such that
\begin{equation}\label{Eq:exponent} G_0(x_1, \ldots, x_d)=\exp(-\nu\{(u_1, \ldots, u_d) \in [0,\infty)^d: \exists i: u_i>x_i \}), \;\;\; (x_1, \ldots, x_d) \in [0,\infty)^d.\end{equation}
This exponent measure is homogeneous of degree $-1$ and there exists a constant $c>0$ such that
\begin{equation}\label{Eq:expmeasure} \nu\{\mathbf{u} \in [0,\infty)^d: \mathbf{u}/\|\mathbf{u}\| \in B, \|\mathbf{u}\|>y \}=c y^{-1} S(B)
\end{equation}
for Borel sets $B \subset \mathbb{S}^{d-1}_+$ and $y>0$.  The same measure $S$ appearing in \eqref{Eq:spectral} and \eqref{Eq:expmeasure} is called the \textit{spectral measure} of $\mathbf{X}$ or $\mathbf{Y}$ and by the above it uniquely describes the dependence structure (or copula) of both exceedances and maxima. 
Due to the marginal standardization of the vector $\mathbf{Y}$ there always exists a constant $c>0$ such that 
 \begin{equation}
 \label{Eq:standard}
 \int_{\mathbb{S}_+^{d-1}}x_i S(d\mathbf{x})=c
 \end{equation} 
 for all $i=1, \ldots, d$. In this analysis, we are less interested in the marginal extremal behavior of individual components and more concerned with the extremal dependence structures.  Hence we are interested in the structure of $S$.
 
 Note that even for a vector $\mathbf{X}$ whose marginals are not in the domain of attraction of an extreme value distribution, it can still make sense to define \eqref{Eq:trans} and look at \eqref{Eq:spectral} 
 in order to describe its extremal behavior. On the other hand, there may also be situations where one does not want to standardize marginals first but treats observations as coming from a random vector $\mathbf{Y}$ which satisfies \eqref{Eq:spectral} with a spectral measure $S$ that does not necessarily satisfy \eqref{Eq:standard}. In any case, the measure $S$ tells us about the angle or direction of an observation that is considered extreme, either in the original scale of observations, or after standardization. If there exist small sets which receive relatively high probabilities under $S$, these sets can be seen as ``typical" directions for an extremal event. The idea of this paper is to identify these sets without assuming a specific underlying model, thereby identifying extremal patterns in a non-parametric way.

We have until now not specified which norm we meant when writing $\|\cdot\|$. The general equivalence between convergence of multivariate maxima in \eqref{Eq:maxlim} and marginal convergences together with convergence of exceedances as described in \eqref{Eq:spectral} holds for any choice of norm $\|\cdot\|$, although the particular choice will of course affect the specific form of the spectral measure $S$. Depending on the particular application, there may be different possible choices for the particular norm. The main results in Section \ref{Sec:main} are therefore formulated in a general way that holds for any choice of norm. For our simulations in Section \ref{Sec:max-linear} and data examples in Section \ref{Sec:data} we use the Euclidean norm $\|\cdot \|_2$.

\subsection{$k$-means and spherical $k$-means}
The $k$-means clustering procedure is a way to identify distinct groups within a population. The name was first introduced in \cite{MacQueen67} although the ideas behind the algorithm dated back further, see \cite{Bock08}. The motivation is to identify cluster centers such that distances of the observations to their nearest cluster centers are minimized. Accordingly, all observations which are closest to the same cluster center are viewed as belonging to the same group. 

In the following, let $d:\mathbb{R}^d \times \mathbb{R}^d \to [0,\infty)$ be a distance function or, more generally, a dissimilarity function in $\mathbb{R}^d$ (see \cite{ganetal07}, Chapter 6). For a probability measure $P$ on $\mathbb{B}(\mathbb{R}^d)$ and a set $A=\{\mathbf{a}_1, \ldots, \mathbf{a}_k\},\ \mathbf{a}_i \in \mathbb{R}^d$ for $i=1, \ldots, k$ and $k \in \mathbb{N}$, one can introduce the averaged distance from any observation to the closest element of $A$ as
\begin{equation}
\label{Eq:Wdist}
W(A, P):= \int \min_{\mathbf{a} \in A} d(\mathbf{x},\mathbf{a})P(d\mathbf{x}) \in [0, \infty]. 
\end{equation} 
For given $P$ and $k$, a set $A_k$ which minimizes $W(A,P)$ among all $A$ with $|A| \leq k$ can be seen as a set of theoretical cluster centers. Note that the set may not necessarily be unique, in which case no clear cluster centers can be identified. 

If we replace $P$ by its sample version $\hat{P}_n$, (i.e. the measure that places mass $1/n$ on each observation $\mathbf{x}_1, \ldots, \mathbf{x}_n$ of a sample) in \eqref{Eq:Wdist}, and derive an accordingly optimal set $\hat{A}_k^n$, its components minimize the sum of the distances from every observation to its nearest cluster center. While the original version of $k$-means uses the Euclidean distance, several alternatives to the algorithm with different choices for $d$ have been suggested. Since our central interest is in the spectral measure and a sample from it has only mass on the unit sphere $\mathbb{S}^{d-1}_+$, it seems natural to measure the distance between two points by their angle, which is then in fact independent of the chosen norm. Therefore we make use of the spherical $k$-means procedure of \cite{DhillonModha01}, which measures differences in terms of angular dissimilarity. We thus set
\begin{equation}\label{Eq:dphi} d(\mathbf{x}, \mathbf{y})=d_{\varphi}(\mathbf{x}, \mathbf{y}):=1-\cos(\mathbf{x},\mathbf{y})=1-\frac{\langle \mathbf{x}, \mathbf{y}\rangle}{\|\mathbf{x}\|_2\|\mathbf{y}\|_2}=1-\frac{\sum_{i=1}^d x_iy_i}{\sqrt{\sum_{i=1}^d x_i^2}\sqrt{\sum_{i=1}^d y_i^2}}, \;\; \mathbf{x}, \mathbf{y} \in \mathbb{R}^d
\end{equation}
or simply $d_{\varphi}(\mathbf{x}, \mathbf{y})=1-\langle\mathbf{x},\mathbf{y}\rangle$ if we restrict $\mathbf{x}, \mathbf{y}$ to the Euclidean unit sphere. Since the optimal cluster centers under this dissimilarity are then determined by their direction only, for $P$ living on $\mathbb{S}_+^{d-1}$ for a chosen norm, the cluster centers will be placed on $\mathbb{S}_+^{d-1}$ by convention, so that we can interpret our cluster centers in the same way as we interpret the projections of extreme observations. Note that in order to allow for more flexibility in the choice of a suitable distance measure, the results in Section \ref{Sec:main} only assume that $d:\mathbb{S}_+^{d-1} \times \mathbb{S}_+^{d-1} \to [0,\infty)$ is a continuous function such that a unique minimizing set in \eqref{Eq:Wdist} exists.

It should be noted that finding the optimal cluster centers for a given $P$ can be an $NP$-hard problem (see \cite{Mahajanetal12}) and the known iterative algorithms often depend crucially on the initial cluster centers, see \cite{BradleyFayyad98}. For the examples and simulations in Sections \ref{Sec:max-linear} and \ref{Sec:data} we rely on the \texttt{R}-package \texttt{skmeans} by \cite{Horniketal12}, which provides short run-times and stable results.

%%%%%%%%%%%%%%%%%%%%%%%%%%%%%%%%%%%%%%%%%%%%%%%

\section{The main result}\label{Sec:main}
In this section we formally introduce our estimation procedure which will, for a given sample, provide a set of empirical cluster centers. Each center can then be interpreted as a ``dependence prototype'' for a particular class of an extremal event. In brief, our procedure looks as follows:
\begin{enumerate}
	\item[1.] With the help of the empirical distribution function, transform a sample from the distribution of $\mathbf{X}$ into (approximately) a sample from $\mathbf{Y}$ as in \eqref{Eq:trans}.
	\item[2.] Choose a fraction of the latter sample that only keeps the transformed observations with largest norm.
	\item[3.] For the chosen subsample, project the transformed observations onto the corresponding unit sphere.
	\item[4.] Apply a spherical $k$-means procedure to the projected observations.
\end{enumerate}
More details about the steps can be found below. Note that steps 1.-3.\ in the above procedure are a way of generating a ``pseudo-sample" from the spectral measure $S$ of standardized observations. But there exist many different ways of statistical inference for $S$ both for standardized and non-standardized data, from non-parametric (e.g.\ \cite{Einmahletal01}, \cite{EinmahlSegers09}), over semiparametric (e.g.\ \cite{Einmahletal97}) to fully parametric procedures (e.g.\ \cite{ColesTawn91}). The following theorem is therefore formulated in a way such that it holds for any estimator of the spectral measure as long as it is weakly or strongly consistent. 

\begin{theorem}\label{Th:main}
	Assume that $S$ is a probability measures on $\mathbb{B}(\mathbb{S}^{d-1}_+)$ and that $S_n, n \in \mathbb{N}$, is a sequence of random probability measures on $\mathbb{B}(\mathbb{S}^{d-1}_+)$ on a common probability space $(\Omega, \mathcal{A}, P)$. Furthermore, assume that $d:\mathbb{S}_+^{d-1} \times \mathbb{S}_+^{d-1} \to [0,\infty)$ is a continuous function.
	
    For each $S_n$ and a given value of $k \in \mathbb{N}$ denote a random set which minimizes 
	$$
	W(A, S_n):= \int_{\mathbb{S}^{d-1}_+} \min_{a \in A} d(x,a)S_n(dx) 
	$$
	among all sets $A$ with at most $k$ elements from $\mathbb{S}^{d-1}_+$ by $A_k^n$. Accordingly, if we replace $S_n$ by $S$, denote the optimal set by $A_k$, and assume that for a given value of $k$, the set $A_k$ is uniquely determined.
	\begin{itemize}
	 \item[a)] If $\int_{\mathbb{S}_+^{d-1}}f(x)S_n(dx) \to \int_{\mathbb{S}_+^{d-1}}f(x)S(dx), n \to \infty,$ in probability for all continuous functions $f:\mathbb{S}_+^{d-1} \to \mathbb{R}$, then $A_k^n$ converges in probability to $A_k$ as $n \to \infty$. 
	 \item[b)] If $\int_{\mathbb{S}_+^{d-1}}f(x)S_n(dx) \to \int_{\mathbb{S}_+^{d-1}}f(x)S(dx), n \to \infty,$ almost surely for all continuous functions $f:\mathbb{S}_+^{d-1} \to \mathbb{R}$, then $A_k^n$ converges almost surely to $A_k$ as $n \to \infty$.
	 	\end{itemize}
\end{theorem}
\begin{remark}\label{rem:theoremremark} In the above theorem, the convergence of sets is formally meant in the Hausdorff distance, but since all involved sets have only finitely many elements, it implies pointwise convergence of elements after a suitable reordering.
\end{remark}	

\begin{proof}[Proof of Theorem \ref{Th:main}] 
	The argumentation is very similar to the one used in \cite{Pollard82}, and the crucial ingredients for the proof are the continuity of $d$ and the compactness of $\mathbb{S}^{d-1}_+$.  First, the assumption a) or b) implies that for a fixed set $A$ we have $W(A,S_n)\to W(A,S)$ in probability or almost surely, respectively, since $d$ is continuous. For a finite number of sets $A^1, \ldots, A^m$ with at most $k$ elements from $\mathbb{S}_+^{d-1}$ the convergence
	\begin{equation}\label{Eq:uniformdisc} \sup_{i=1, \ldots, m}|W(A^i,S_n)-W(A^i,S)| \to 0, \;\;\; n \to \infty, 
	\end{equation}
	in the respective mode of convergence follows then immediately. Due to continuity of $d$ and its compact support we can furthermore find for each $\delta>0$ an $m \geq 1$ and sets $A^1, \ldots, A^m$ with at most $k$ elements from $\mathbb{S}_+^{d-1}$ such that 
	$$ \min_{i=1, \ldots, m} \sup_{x \in \mathbb{S}_+^{d-1}} |\min_{a \in A} d(x,a)-\min_{a \in A^i}d(x,a)|<\delta $$
        for all $A$ with at most $k$ elements from $\mathbb{S}_+^{d-1}$. This in turn implies that 
        \begin{equation*}\label{Eq:uniformcont} \sup_{A \subset \mathbb{S}_+^{d-1}, |A| \leq k} \min_{i=1, \ldots, m} |W(A,P)-W(A^i,P)|<\delta
        \end{equation*}
        for all probability measures $P$ on $\mathbb{B}(\mathbb{S}_+^{d-1})$ and therefore
        $$ \sup_{A \subset \mathbb{S}_+^{d-1}, |A| \leq k}|W(A,S_n)-W(A,S)|\leq \sup_{i=1, \ldots, m}|W(A^i,S_n)-W(A^i,S)|+2 \delta. $$
        Convergence \eqref{Eq:uniformdisc} then implies
	$$ \sup_{A \subset \mathbb{S}_+^{d-1}, |A| \leq k}|W(A,S_n)-W(A,S)| \to 0 $$
	in probability or almost surely, respectively. Thereby we also have $|W(A_k^n,S_n)-W(A_k^n,S)|\to 0$ and thus $W(A_k^n,S_n) \to W(A_k,S)$ for our optimal sets $A_k^n$ and $A_k$.
	Due to the continuity of $d$ and its compact support, the assumption of a unique optimal set $A_k$ implies that for each $\delta>0$ there exists an $\epsilon>0$ such that $|W(\tilde{A},S)-W(A_k,S)|<\epsilon$ implies $d_{\mbox{\scriptsize Hausdorff}}(\tilde{A}, A_k)<\delta$. This finally implies that $d_{\mbox{\scriptsize Hausdorff}}(A_k^n, A_k) \to 0$ in the respective mode of convergence.  
\end{proof}

Since the idea of our approach is to detect general patterns without relying on a particular model, we choose for the rest of the analysis a straightforward non-parametric estimator of the spectral measure of standardized observations which is a natural empirical counterpart to \eqref{Eq:spectral} and a slight modification of the estimator introduced in \cite{Einmahletal01}. 

Note that the spectral measure of $\mathbf{X}$ is defined in terms of the random vector $\mathbf{Y}$ from \eqref{Eq:trans}, but that we do not know the marginal distributions $F_i, 1 \leq i \leq  d$. A solution to this is to replace $F_i$ by the left-continuous version of the empirical distribution function 
 $$ F_{i,n}(x)=\frac{1}{n}\sum_{j=1}^n \mathds{1}_{\{X_i^j < x\}}, \;\;\; x \in \mathbb{R}, $$ 
 and transform the observations to
 \begin{equation}\label{Eq:ranktransform} \hat{\mathbf{Y}}_j=(\hat{Y}_1^j, \ldots, \hat{Y}_d^j), \;\;\; \mbox{ with } \hat{Y}_i^j := (1-F_{i,n}(X_i^j))^{-1}.
 \end{equation}

 The empirical counterpart of \eqref{Eq:spectral} then motivates the estimator
  \begin{equation}\label{Eq:empspec} \hat{S}_n(B):=\frac{\sum_{j=1}^n \mathds{1}_{\left\{\|\hat{\mathbf{Y}}_j\| \geq \frac{n}{l_n},\frac{\hat{\mathbf{Y}}_j}{\|\hat{\mathbf{Y}}_j\|} \in B\right\}}}{\sum_{j=1}^n \mathds{1}_{\left\{\|\hat{\mathbf{Y}}_j\| \geq \frac{n}{l_n} \right\}}},
  \end{equation}
  for Borel subsets $B$ of $\mathbb{S}^{d-1}_+$, where $l_n \in \mathbb{N}$. Note that due to their definition the observed components of the $\hat{\mathbf{Y}}_i's$ will have values in $\{1,n/(n-1), \ldots, n/2, n\}$. Therefore, the value of $l_n$ should grow with $n$ in order to obtain a consistent estimator, but the growth rate should also not be too fast in order to catch only the extremal observations. Proposition \ref{Prop:ourest} below gives necessary assumptions on $l_n$ for the weak and strong consistency of this estimator.  In a second step, for a chosen value of $l_n$, we can then apply a spherical $k$-means algorithm to $\hat{S}_n$, i.e.\ to the selected projections of extreme observations onto the unit sphere. The following proposition shows that the resulting estimators are consistent.

\begin{proposition}\label{Prop:ourest}
	Assume that $\mathbf{X}_i=(X_1^i, \ldots, X_d^i), i \in \mathbb{N},$ are i.i.d.\ copies of a vector $\mathbf{X}$, such that the transformed vector $\mathbf{Y}$ from \eqref{Eq:trans} satisfies \eqref{Eq:spectral} with spectral measure $S$. For $S_n=\hat{S}_n$ as in \eqref{Eq:empspec}, define the sets $A_k^n$ and $A_k$ as in Theorem \ref{Th:main} and assume that the set $A_k$ is uniquely determined for a given value of $k$. Then
	\begin{itemize}
	\item[a)] if $l_n/n \to 0$ and $l_n \to \infty$ the sets $A_k^n$ converge to $A_k$ in probability;
	\item[b)] if $l_n/n \to 0$ and $l_n/\log(\log(n)) \to \infty$ the sets $A_k^n$ converge to $A_k$ almost surely.
	\end{itemize}
	
\end{proposition}
\begin{proof}
The proposition follows from Theorem \ref{Th:main} if we can verify that the sequence of random measures $S_n$ satisfies the corresponding assumptions. To see this, we argue similar to \cite{Einmahletal01}, Theorem 1. We start by looking at the estimator
$$ \hat{l}(x_1, \ldots, x_d):=\frac{1}{l_n}\sum_{i=1}^n \mathds{1}_{\{\exists j: \hat{Y}_j^i > \frac{n}{l_n x_j}\}} $$
for the so-called stable tail dependence function
\begin{equation}
\label{Eq:stdf}
l(x_1, \ldots, x_d):=\nu\{(u_1, \ldots, u_d) \in [0,\infty)^d: \exists j: u_j>1/x_j\}, 
\end{equation} 
with $\nu$ as introduced in \eqref{Eq:exponent}. For $l_n/n \to 0, l_n \to \infty$ the estimator converges in probability (see \cite{Huang92}) and for $l_n/n \to 0, l_n/\log(\log(n)) \to \infty$ it converges almost surely (see \cite{Qi97}) for all $x_1, \ldots, x_d \in (0,\infty)^d$. This pointwise convergence can be extended to show that
$$ \frac{1}{l_n}\sum_{i=1}^n \mathds{1}_{\left\{\|\hat{\mathbf{Y}}_i\|_\infty>\frac{n}{l_n c}, \,\hat{\mathbf{Y}}_i\frac{l_n}{n} \in \hat{B} \right\}} \to \nu\left\{\mathbf{u} \in [0,\infty)^d: \|\mathbf{u}\|_\infty>1/c, \mathbf{u} \in \hat{B} \right\}$$
for each $c>0$ and continuity set $\hat{B} \subset [0,\infty)^d$ in the respective mode of convergence, see the proof of Theorem 1 in \cite{Einmahletal01} for details. Due to equivalence of norms there exists a $c_0>0$ such that $\|\mathbf{x}\|>1$ implies $\|\mathbf{x} \|_{\infty}>1/c_0$ for the chosen norm $\| \cdot \|$. Set now $\hat{B}=\{\mathbf{u} \in [0,\infty)^d: \|\mathbf{u}\|>1, \mathbf{u}/\|\mathbf{u}\| \in B\}$ for an $S$-continuity set $B \subset \mathbb{S}^{d-1}_+$ and $A=\{\mathbf{u} \in [0,\infty)^d: \|\mathbf{u}\|>1\}$ so that
$$ \hat{S}_n(B)=\frac{\sum_{i=1}^n \mathds{1}_{\left\{\|\hat{\mathbf{Y}}_i\|_\infty>\frac{n}{l_n c_0}, \,\hat{\mathbf{Y}}_i\frac{l_n}{n} \in \hat{B} \right\}}}{\sum_{i=1}^n \mathds{1}_{\left\{\|\hat{\mathbf{Y}}_i\|_\infty>\frac{n}{l_n c_0}, \,\hat{\mathbf{Y}}_i\frac{l_n}{n} \in A \right\}}}\to \frac{\nu\left\{\mathbf{u} \in [0,\infty)^d: \|\mathbf{u}\|_\infty>1/c_0, \mathbf{u} \in \hat{B} \right\}}{\nu\left\{\mathbf{u} \in [0,\infty)^d: \|\mathbf{u}\|_\infty>1/c_0, \mathbf{u} \in A \right\}}=S(B)$$
again in the respective mode of convergence. This implies the weak convergence either in probability or almost surely of $\hat{S}_n$ to $S$ (see again \cite{Einmahletal01}, Theorem 1) and thereby that the assumption a) or b), respectively, of Theorem \ref{Th:main} is satisfied.
\end{proof}

%%%%%%%%%%%%%%%%%%%%%%%%%%%%%%%%%%%%%%%%%%%%%%%

\section{Application to max-linear models}\label{Sec:max-linear}

The idea behind the decomposition of a spectral measure into $k$ clusters is motivated by the special case where the spectral measure is clearly concentrated around $k$ different centers.  An idealized example is provided by the max-linear model where a spectral measure has only masses on $k$ different points.  In the following we will demonstrate how our $k$-means procedure can be seen as an alternative way of estimating their model parameters.

A max linear model consists of $k$ different so-called \textit{factors} $\mathbf{a}_i=(a_1^i, \ldots, a_d^i) \in [0,\infty)^d, i=1, \ldots, k$ from which a random vector is generated by
\begin{equation}\label{Eq:maxlin} (X_1, \ldots, X_d)=\left(\max_{i=1, \ldots, k} a_1^i Z_i, \ldots, \max_{i=1, \ldots, k}a_d^i Z_i\right), 
\end{equation}
where $Z_1, \ldots, Z_k$ are i.i.d.\ random variables with the same heavy-tailed distribution. The most common choice for this distribution is a standard Fr\'{e}chet-distribution. Furthermore, one typically assumes that $\sum_{i=1}^k a_j^i=1$ for all $j=1, \ldots, d$ such that all margins of $\mathbf{X}$ are standard Fr\'{e}chet as well. 
 
 Looking at \eqref{Eq:maxlin} it is clear that the largest observations of $\mathbf{X}$ are due to a large observation of a $Z_i$ and therefore the factors $\mathbf{a}_i$ determine the possible directions of extremal observations. In fact one can show that the spectral measure $S$ concentrates on the points $\mathbf{s}_i=\mathbf{a}_i/\|\mathbf{a}_i\|$ with corresponding probabilities $p_i=\|\mathbf{a}_i\|/(\sum_{j=1}^k\|\mathbf{a}_j\|), 1 \leq i \leq k$.  On the other hand, for each discrete spectral measure with mass concentrated on $k$ points there exists a max-linear model with $k$ factors which results in this given spectral measure, see \cite{YuenStoev14}. It is also shown that any given dependence structure of extremes can be approximated arbitrarily well by spectral measures generated from max-linear models if one allows the number of factors to grow, see \cite{Fougeresetal13}. Furthermore, it was recently shown in \cite{Gissibl18, Gissibletal18, GissiblKlueppelberg18} that max-linear models also evolve from a natural modeling of extremal dependence generated from a directed acyclic graph of components, thereby allowing the modeling and detection of causality in extreme events. 
 
 Parameter estimation of max-linear models has proven to be a difficult task and previous approaches to estimate model parameters based on extremal observations had to address the fact that no spectral density exists which excludes standard maximum likelihood procedures. Instead, \cite{Einmahletal12}, \cite{Einmahletal16} and \cite{Einmahletal18} use a least squares estimator based on the stable tail dependence function to estimate parameters from extremal observations and \cite{YuenStoev14} construct a least squares estimator based on the joint distribution function and make use of all observations. In the following we illustrate how the $k$-means procedure serves as an alternative and effective way of inference for the max-linear models.
 
 From the above it follows that the discrete spectral measure, i.e., the points $\mathbf{s}_1, \ldots, \mathbf{s}_k \in \mathbb{S}^{d-1}_+$ on which the spectral measure $S$ is concentrated, and the corresponding probabilities $p_1, \ldots, p_k$, are an equivalent way of parametrizing a max-linear model. For such a spectral measure, it is clear that $W(A,S)$ as defined in \eqref{Eq:Wdist} is minimized by choosing $A=\{\mathbf{s}_1, \ldots, \mathbf{s}_k\}$ and $A$ is uniquely determined. For an estimator $S_n$ of $S$, which satisfies the assumptions from the previous chapter, the $k$-means cluster centers can therefore be seen as consistent estimators of $\mathbf{s}_1, \ldots, \mathbf{s}_k$. This consistency also implies that the percentages of points which are classified as belonging to cluster $i$ converge to the corresponding probability $p_i, 1 \leq i \leq k$. Especially if one is interested in using the max-linear model as an approximation for the largest observations only, the estimation of the $\mathbf{s}_i$'s and $p_i$'s can be seen as an alternative to the estimation of the components $\mathbf{a}_i$'s. 
 
 In order to compare the spherical $k$-means procedure with the previously mentioned approaches we set up a small simulation study. To this end, we first generate a random parameter constellation for a max-linear model with $d$ dimensions and $k$ factors. We then estimate the factor coefficients according to \cite{Einmahletal16} and \cite{Einmahletal18}, as provided by the $\texttt{R}$-package \cite{Kiriliouk16}, and \cite{YuenStoev14}, as provided by \cite{Yuen15}, and derive the resulting spectral measure. We then compare the estimators for $\mathbf{s}_1, \ldots, \mathbf{s}_k$ and $p_1, \ldots, p_k$ as derived from the estimated parameters of the max-linear models on one hand and on the other hand as derived directly by the $k$-means estimator $S_n^k$ which puts mass $\hat{p}_i$ in the cluster center $\hat{s}_i$, where 
 $$ \hat{p}_i=S_n(\{\mathbf{x} \in \mathbb{S}^{d-1}_+: \argmin_{j=1, \ldots, k}d_\varphi(\mathbf{x}, \mathbf{s}_j)=i\}), \;\;\; i=1, \ldots, k,$$
 i.e.\ the percentage of observations that are classified as belonging to cluster $i$. 
 The difference between the true spectral measure $S$ and the corresponding estimator is evaluated by two different evaluation criteria. For the estimated cluster centers, we first derive 
  $$ d_s(S_n,S):=\min_{\pi: \pi\, \mbox{\tiny is permutation of $\{1, \ldots, k\}$}}\sqrt{\sum_{j=1}^k\|\hat{s}_{\pi(j)}-s_j\|_2^2},$$
  which can be seen as a distance measure similar to the Hausdorff distance applied to finite sets, but taking into account all distances of the matched vectors instead of only the maximal one. This gives an idea about how well the estimator identifies possible extremal directions, but does not take into account their frequencies. To this end we also look at a metric on the space of probability measures, where we use the Wasserstein metric with $p=1$, which is defined as
  $$  W_1(S_n, S):= \inf_{P \in \Gamma(S,S_n)} \int_{\mathbb{S}^{d-1}_+\times\mathbb{S}^{d-1}_+} \|\mathbf{x}-\mathbf{y}\|_2 P(\mbox{d}\mathbf{x},\mbox{d}\mathbf{y}), $$
  where $\Gamma(S,S_n)$ is the set of all probability measures on $\mathbb{S}^{d-1}_+\times\mathbb{S}^{d-1}_+$ with first marginal $S$ and second marginal $S_n$. We use the \texttt{R}-package \texttt{transport}, see \cite{Schuhmacheretal19}, for evaluation of Wasserstein distances.
 
 We first assume $k$ to be known. For each combination of $d$ and $k$ we randomly generate 100 model specifications and for each model specification we generate 1000 observations. The random factors are generated as below, where $U_i, i \in \mathbb{N},$ stand for i.i.d.\ random variables with uniform distribution on $[0,1]$. We only state the first $k-1$ factors, as the last factor is always determined from the first ones by the standardization assumption.
 \begin{itemize}
  \item $d=4, k=2$: First factor is $(U_1,U_2,U_3,U_4)/2$.
  \item $d=4, k=6$: First five factors are $(U_1, U_2, U_3, U_4)/3$, $(U_5, 0, U_6, 0)/3$, $(0, U_7, 0, U_8)/3$, \linebreak $(U_9, U_{10}, 0, 0)/3$, $(0, 0, U_{11}, U_{12})/3$.
  \item $d=k=6$: First five factors are $(U_1, \ldots, U_6)/3$, $(0, U_7, 0, U_8, 0, U_9)/3$, $(U_{10}, 0, U_{11}, 0, U_{12},0)/3$, $(0, 0, 0, U_{13}, U_{14}, U_{15})/3$, $(U_{16}, U_{17}, U_{18}, 0, 0, 0)$.  
  \item $d=10, k=6$: First five factors are $(U_1, \ldots, U_{10})/2, (U_{11}, U_{12}, 0, \ldots , 0)/2$,\newline $(0, 0, U_{13}, U_{14},  0, \ldots , 0)/2$, $(0,0,0,0, U_{15}, U_{16}, 0, 0, 0, 0)/2$, $(0, \ldots, 0, U_{17}, U_{18}, U_{19}, U_{20})/2$. 
 \end{itemize}

 For the method from \cite{YuenStoev14} we use all observations, for the remaining only the 100 with largest norm. The grid for the estimator from \cite{Einmahletal18} includes all $d$-dimensional vectors with entries from the set $\{0, 1/3, 2/3, 1\}$ and exactly 2 non-zero entries. A finer grid would have implied very long run times. For those three estimators, the fact that there are 0's in the factors and knowledge of their positions is not used for the estimation, so there are $d \cdot (k-1)$ parameters to estimate. 

 \begin{table}[ht]
 \begin{center}
  \begin{tabular}{c | c c c c}
  Method & $d=4, k=2$ & $d=4, k=6$& $d=k=6$ & $d=10, k=6$\\
  \hline

 (spherical) $k$-means & 0.0562 (0.0276) & 0.4042 (0.2693) & 0.3376 (0.2050)
 & 0.3711 (0.2306) \\
 \cite{Einmahletal16} & 0.0728 (0.0415) & 0.5255 (0.3198) & 0.4295 (0.2506) & 0.4264 (0.2694) \\
 \cite{Einmahletal18} & 0.0605 (0.0309) & 0.4868 (0.3001) &  0.3971 (0.2271) & 0.4103 (0.2312) \\
 \cite{YuenStoev14} & 0.1370 (0.1905) & 1.1049 (0.2052) &  1.3808 (0.1541)  & 1.8755 (0.1582) \\
\end{tabular}
\caption{Comparison of $d_s(S_n,S)$ for different values of $d$ and $k$. Mean over 100 simulations, standard deviation in brackets.}
\label{Tab:dsvalues}
\end{center}
 \end{table}
 
  \begin{table}[ht]
 \begin{center}
  \begin{tabular}{c | c c c c}
  Method & $d=4, k=2$ & $d=4, k=6$& $d=k=6$ & $d=10, k=6$\\
  \hline
 (spherical) $k$-means & 0.0450 (0.0206) & 0.1310 (0.0285) & 0.1393
 (0.0302)
 & 0.1578 (0.0332) \\
 \cite{Einmahletal16} & 0.0629 (0.0351) &  0.1445 (0.0344) &  0.1464 (0.0328) & 0.1592 (0.0345) \\
 \cite{Einmahletal18} & 0.0458 (0.0219) &  0.1310 (0.0314) & 0.1386 (0.0329) & 0.1569 (0.0348) \\
 \cite{YuenStoev14} & 0.0725 (0.0977) & 0.3336 (0.0876) &  0.4860 (0.0625) &  0.6764 (0.0648) \\
\end{tabular}
\caption{Comparison of $W_1(S_n,S)$ for different values of $d$ and $k$. Mean over 100 simulations, standard deviation in brackets.}
\label{Tab:wasservalues}
\end{center}
 \end{table}
 
 The average values of $d_s(S_n,S)$ and $W_1(S_n,S)$ over all 100 realizations for different constellations are shown below in Tables~\ref{Tab:dsvalues} and \ref{Tab:wasservalues}, with the observed standard deviations in brackets. It can be seen that the locations of the points of mass of the spectral measure are most precisely estimated by the spherical $k$-means procedure. The accuracy in estimating the spectral measure is very similar for all methods except for the CRPS-method from \cite{YuenStoev14}, which was constructed for an overall good fit but with little focus on extremes. 
 
Since the number of factors is usually not known a priori, we also look at two misspecified models, where in both cases the model was fitted as if there were $k=3$ factors, but the true value of $k$ is 2 or 6. The random models are generated as described previously for the respective constellations of $d$ and $k$. 
   \begin{table}
 \begin{center}
  \begin{tabular}{c | c c }
  Method & $d=4, k=2$ & $d=4, k=6$\\
  \hline
 (spherical) $k$-means & 0.0504 (0.0226)& 0.2746 (0.0537) \\
 \cite{Einmahletal16} & 0.0571 (0.0270) & 0.3225 (0.0646) \\
 \cite{Einmahletal18} & 0.0532 (0.0230) & 0.2921 (0.0663) \\
 \cite{YuenStoev14} &  0.0755 (0.0649)  & 0.4230 (0.1039) \\  
\end{tabular}
\caption{Comparison of $W_1(S_n,S)$ for true model parameters $d=4, k=2$ and $d=4, k=6$ but with models fitted to $k=3$ instead. Mean over 100 simulations, standard deviation in brackets.}
\label{Tab:misspecwasservalues}
\end{center}
 \end{table}
 In Table~\ref{Tab:misspecwasservalues} we see the results for the two misspecified models as measured in the $W_1$-metric of estimated and true spectral measure. In the two examples, the spherical $k$-means procedure copes better with the fact that our model is misspecified. 
 
 Regarding the numerical implementation of the estimators \cite{Einmahletal16} and \cite{Einmahletal18} we noted in our simulations that the first depends for larger values of $d$ and $k$ heavily on the starting parameters of the algorithm, where we choose the starting value for factor parameters from the spherical $k$-means estimator, as also suggested in \cite{Einmahletal12}. The procedure for the estimator from \cite{Einmahletal18} has very long runtimes, even with the relatively coarse grid that we use. 
 
 We conclude from the simulations that the spherical $k$-means procedure is, for the chosen examples and metrics, superior or competitive to (and usually faster than) methods for estimating max-linear models if one is mainly interested in the resulting spectral measure of observations.

%%%%%%%%%%%%%%%%%%%%%%%%%%%%%%%%%%%%%%%%%%%%%%%

\section{Data examples}\label{Sec:data}

In the following we apply our procedure to three different data sets with dimensions 5, 30 and 38. For all data sets we observe that in each estimated cluster center there are many components with values close to 0 and only a few which have significantly positive entries. This hints at the phenomenon of asymptotic independence. In extreme value theory, two random variables $X,Y$ with distributions $F_X, F_Y$ are called \textit{asymptotically independent} if
 $$ \lim_{u \nearrow 1}P(X>F_X^{-1}(u)) | Y>F_Y^{-1}(u))=0,$$
(assuming that the limit exists) and \textit{asymptotically dependent} else. Detecting the groups of random variables which share asymptotic dependencies and classify them accordingly was the main aim of \cite{Chautru15}. Our analysis allows for a more gradual view on dependencies since looking at the estimated clusters and observed differences in cluster components gives an idea about strong and weak hints towards asymptotic dependence or independence. For groups of asymptotically dependent variables it furthermore allows to identify patterns within those groups.

\subsection{Air pollution data}
We start the analysis with a dataset of relatively small dimension, namely the air pollution data which has been analyzed in \cite{HeffernanTawn04}. It consists of daily measurements of five air pollutants in the city center of Leeds (U.K.), collected between 1994 and 1998 and split up into summer and winter months which gives 578 and 532 observations, respectively. The data is available via the \texttt{R}-package \texttt{texmex}, see \cite{texmex18}.  We apply the four-step procedure as outlined at the beginning of Section \ref{Sec:main}, with the transformation of marginals as described in \eqref{Eq:ranktransform}. For this data set, we use the 10\% of transformed observations with largest Euclidean norm, project them on the unit sphere and apply the spherical $k$-means procedure. For this and the other two data examples we use the command \texttt{skmeans} with method \texttt{pclust} and additional parameters \texttt{nruns = 1000, maxchains=100} from \cite{Horniketal12}.

\begin{figure}[H]
	\begin{centering}
	\includegraphics[height=6cm]{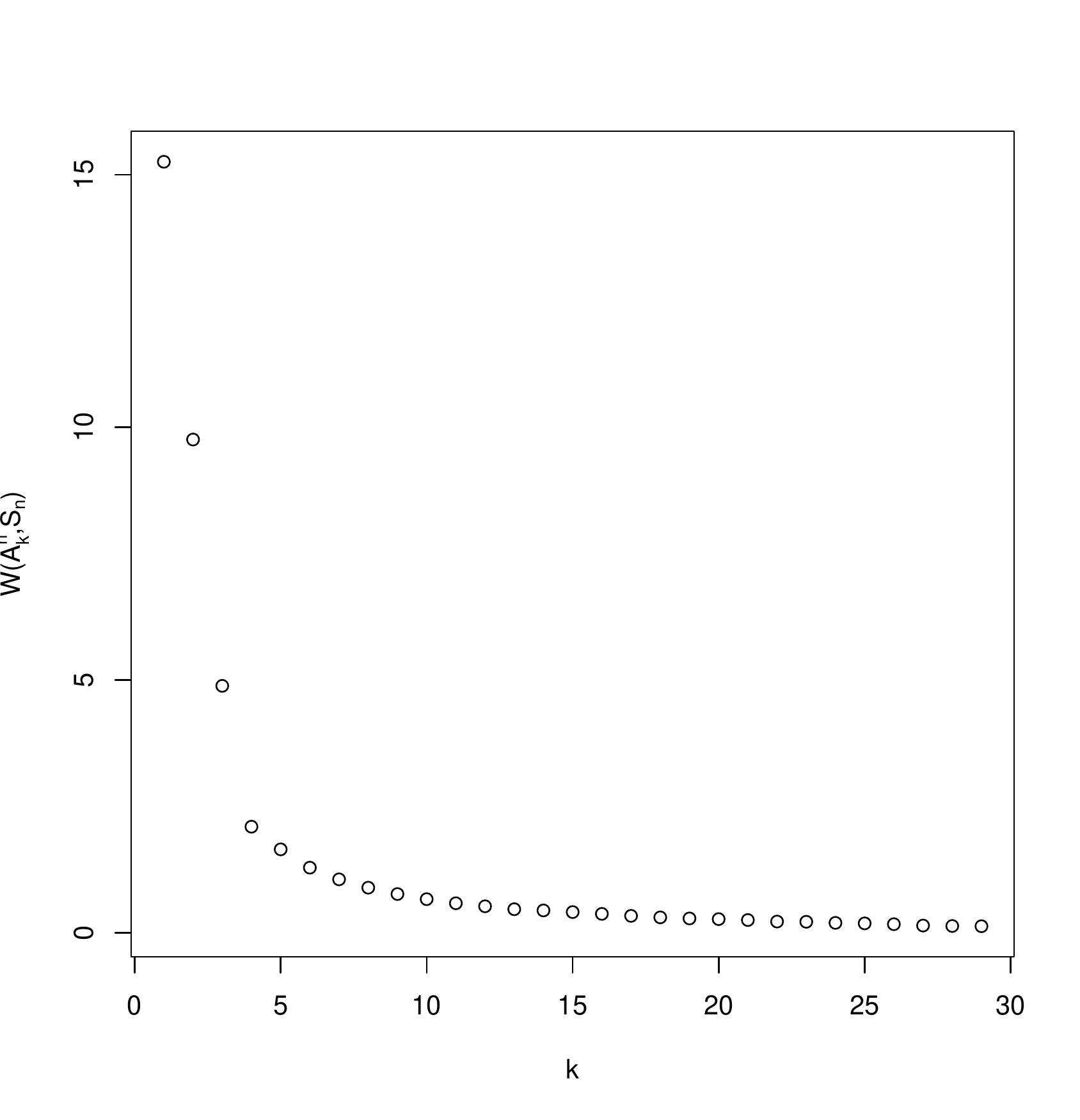}
	\includegraphics[height=6cm]{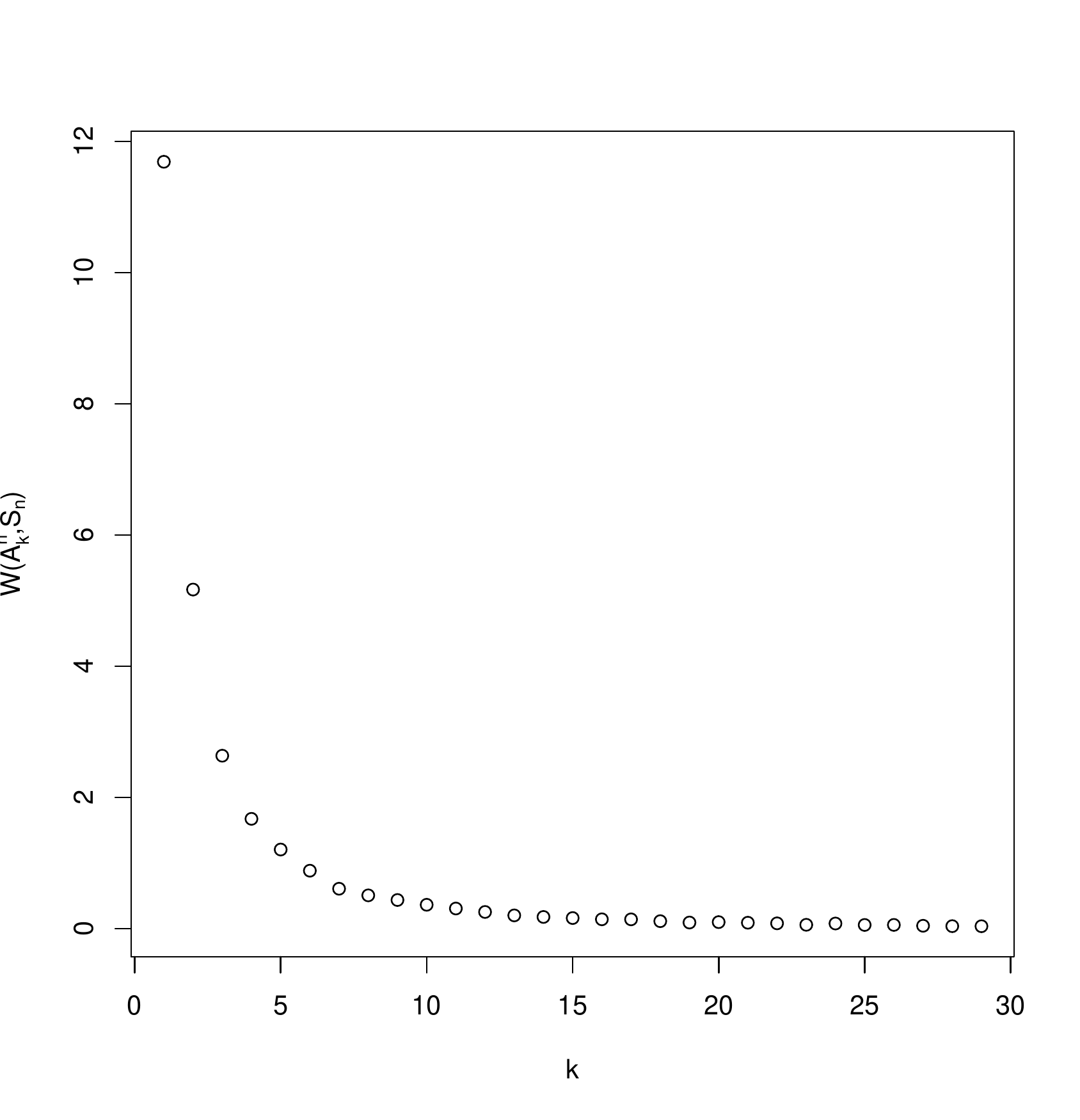}
	\caption{The value of $W(A_k^n,S_n)$ for different values of $k$ in the air pollution data sets. Left: summer data. Right: winter data.}
	\label{Fig:airpollution1}
	\end{centering}
\end{figure}

The first step is now to determine suitable values of $k$. A common way of doing this is creating a so-called ``elbow plot" by plotting the minimized distance $W(A_k,S_n)$ (recall from \eqref{Eq:Wdist}) against the number of clusters $k$, see Figure~\ref{Fig:airpollution1} for the summer and winter data. Note that $W(A_k,S_n)$ necessarily decreases with the increase of $k$. In the plot one usually looks for a $k$ such that for larger values the decrease becomes insignificant, but we note that there is no clear theoretical criterion for an optimal choice of $k$. Our goal is to use the algorithm as a tool to explore the structure in the data and we stress that it often makes sense to look at different values of $k$ in the analysis.

From the elbow plots of the summer and winter data we decide to pursue our analysis for both $k=4$ and $k=5$. The graphs in Figure~\ref{Fig:airpollution2} and Figure~\ref{Fig:airpollution3} illustrate the cluster centers for $k=4$ and $k=5$, respectively.  In order to provide visual comparison, we re-normalize each cluster center such that the maximum component is scaled up to 1, i.e.,
$$
(c_1^{(j)},\ldots,c_d^{(j)}) \rightarrow \left(\frac{c_1^{(j)}}{\max_h\{c_h^{(j)}\}},\ldots,\frac{c_d^{(j)}}{\max_h\{c_h^{(j)}\}}\right).
$$

\begin{figure}[h]
	\begin{center}
	\includegraphics[height=7cm]{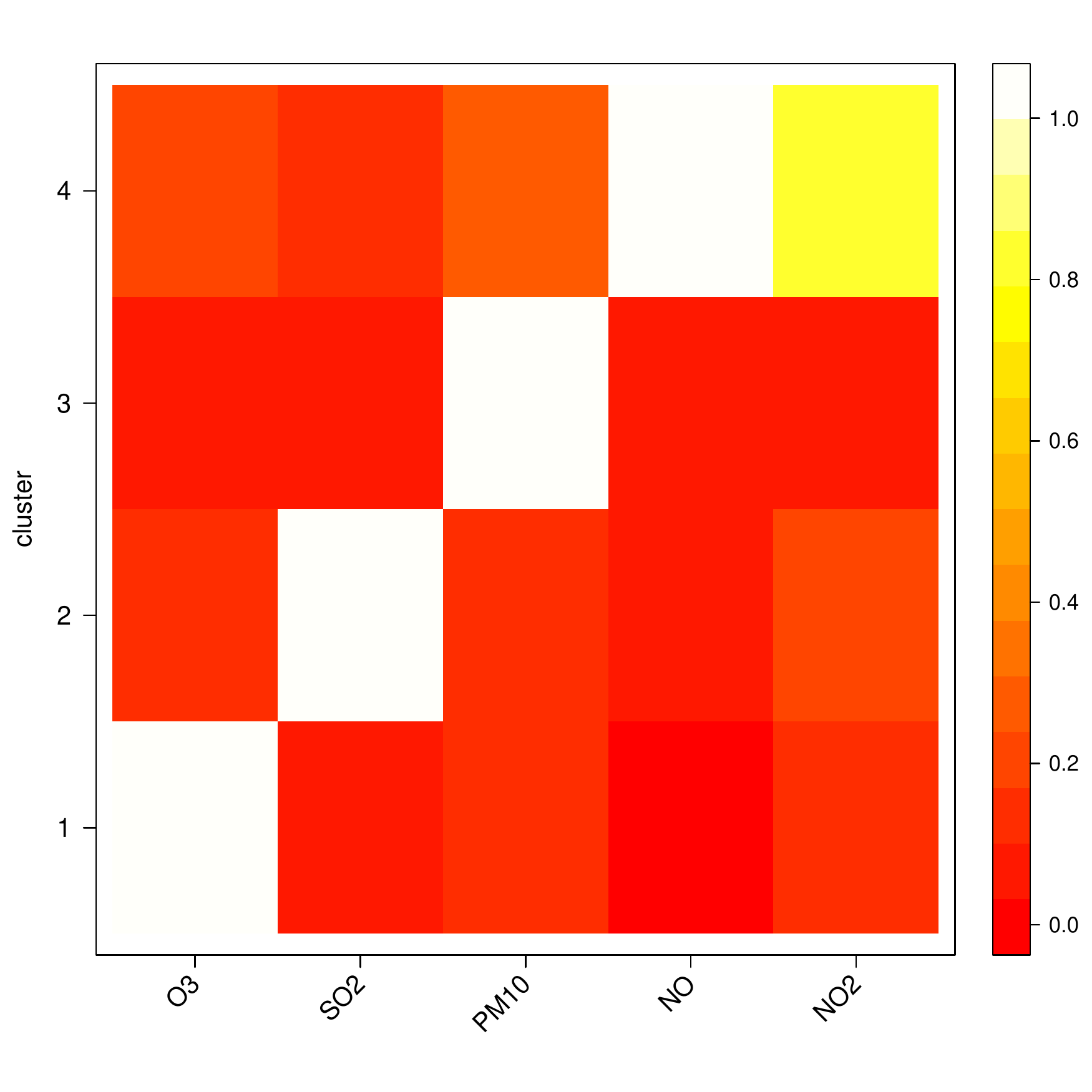}
	\includegraphics[height=7cm]{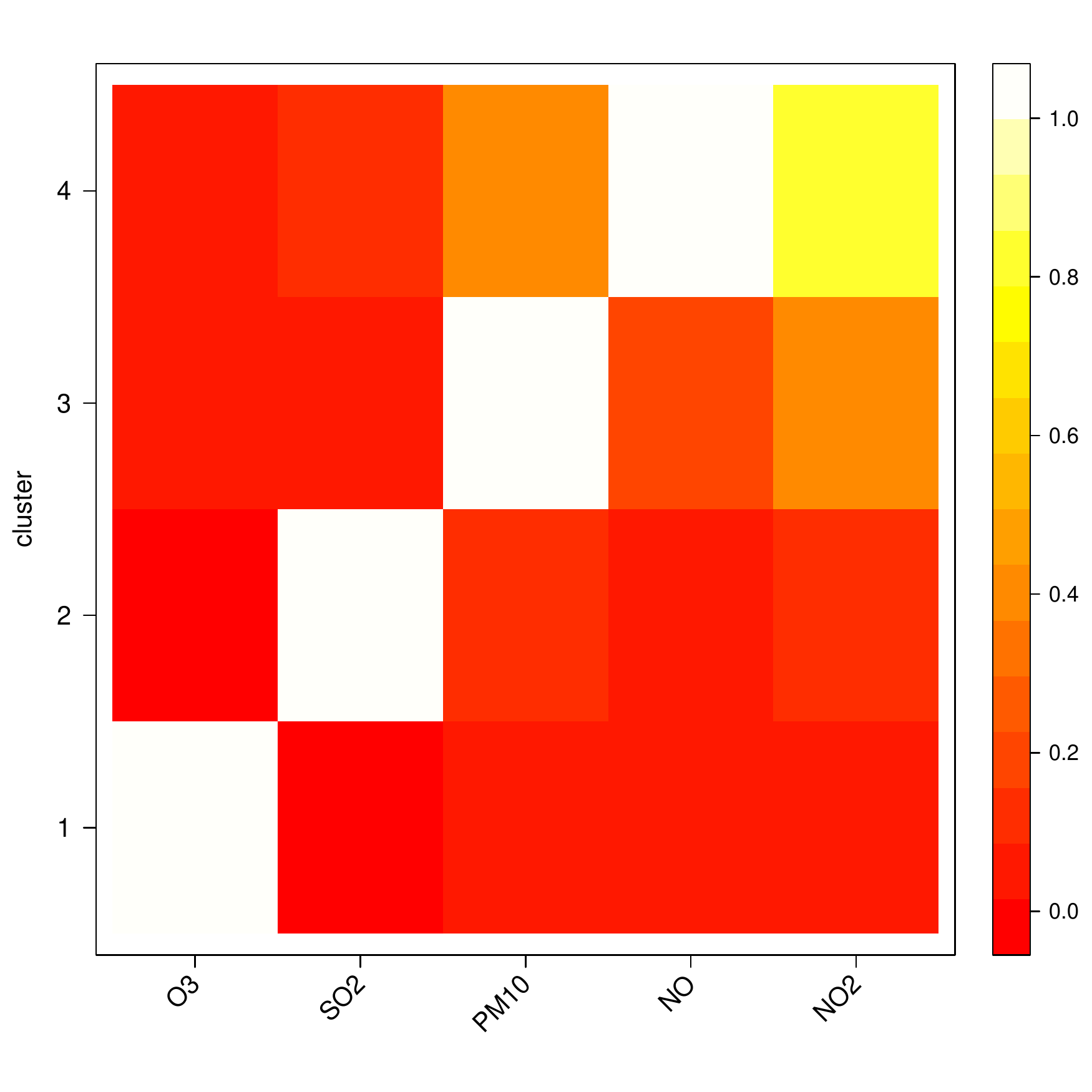}
	\caption{The $k$-means clustering result on the air pollution data for $k=4$. Left: summer data. Right: winter data.}
	\label{Fig:airpollution2}
    \end{center}
	\end{figure}
	\begin{figure}[h]
    \begin{center}
	\includegraphics[height=7cm]{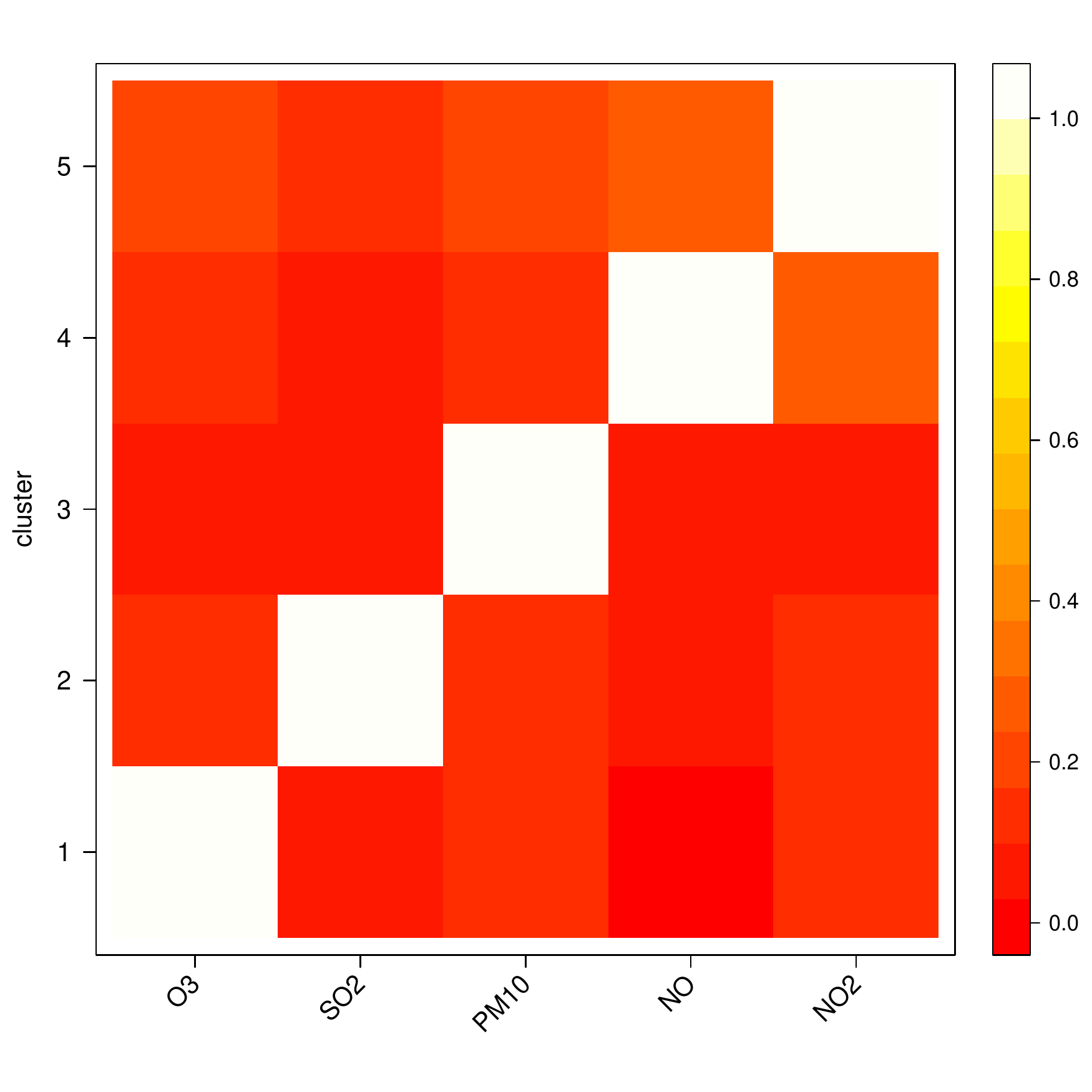}
	\includegraphics[height=7cm]{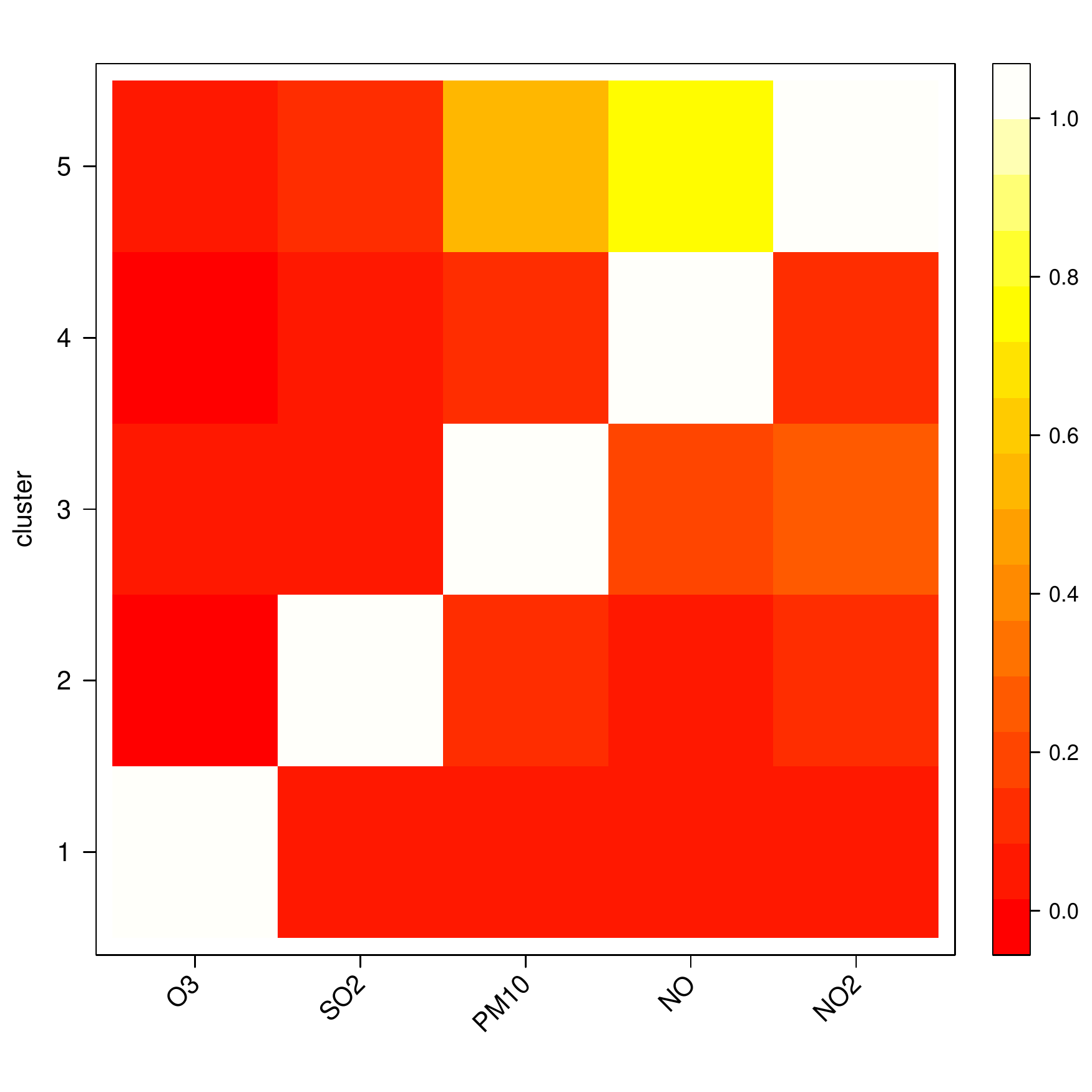}
	\caption{The $k$-means clustering result on the air pollution data for $k=5$. Left: summer data. Right: winter data.}
	\label{Fig:airpollution3}
	\end{center}
\end{figure}
 For $k=5$ we see that for the summer data each cluster has exactly one large component, hinting at asymptotically independence between all air pollutants. We note that due to the marginal transformation in \eqref{Eq:ranktransform} and the standardization property \eqref{Eq:standard} of the spectral measure $S$, the components of our projections on the unit sphere should all have approximately the same expected value. Therefore, each component should correspond to a large entry in at least one cluster center. As a result, we note that for $k=4$ there has to be at least one cluster where at least two components are large. We can interpret cluster 4 in the summer data for $k=4$ in the way that NO and NO$_2$ are the air pollutants which are most likely to occur at extreme levels simultaneously. Both for $k=4$ and for $k=5$ we can identify a tendency for particulate matter (PM$_{10}$), NO and NO$_2$ to occur jointly at extreme levels, but only in winter. This is in line with the conclusions in \cite{HeffernanTawn04} who state asymptotic independence for all components except for PM$_{10}$, NO and NO$_2$ in winter.

\subsection{Financial portfolio losses} \label{Subsec:portfolio}

In this example, we consider the `value-averaged' daily returns of 30 industry portfolios compiled and posted
as part of the Kenneth French Data Library.  The data in consideration span between 1950--2015 with $n=16694$ observations.  This is the same dataset as analyzed in \cite{CooleyThibaud16}, where dependencies in extreme losses were explored with a method related to principle component analysis. Here we attempt to recover more information using our method. Since we are interested in extremal losses we first multiply all returns by -1. After that we use the same procedure as for the previous data set, but this time we only look at the transformed observations with the largest 5\% of Euclidean norms.

\begin{figure}[ht]
	\centering
	\includegraphics[height=5cm]{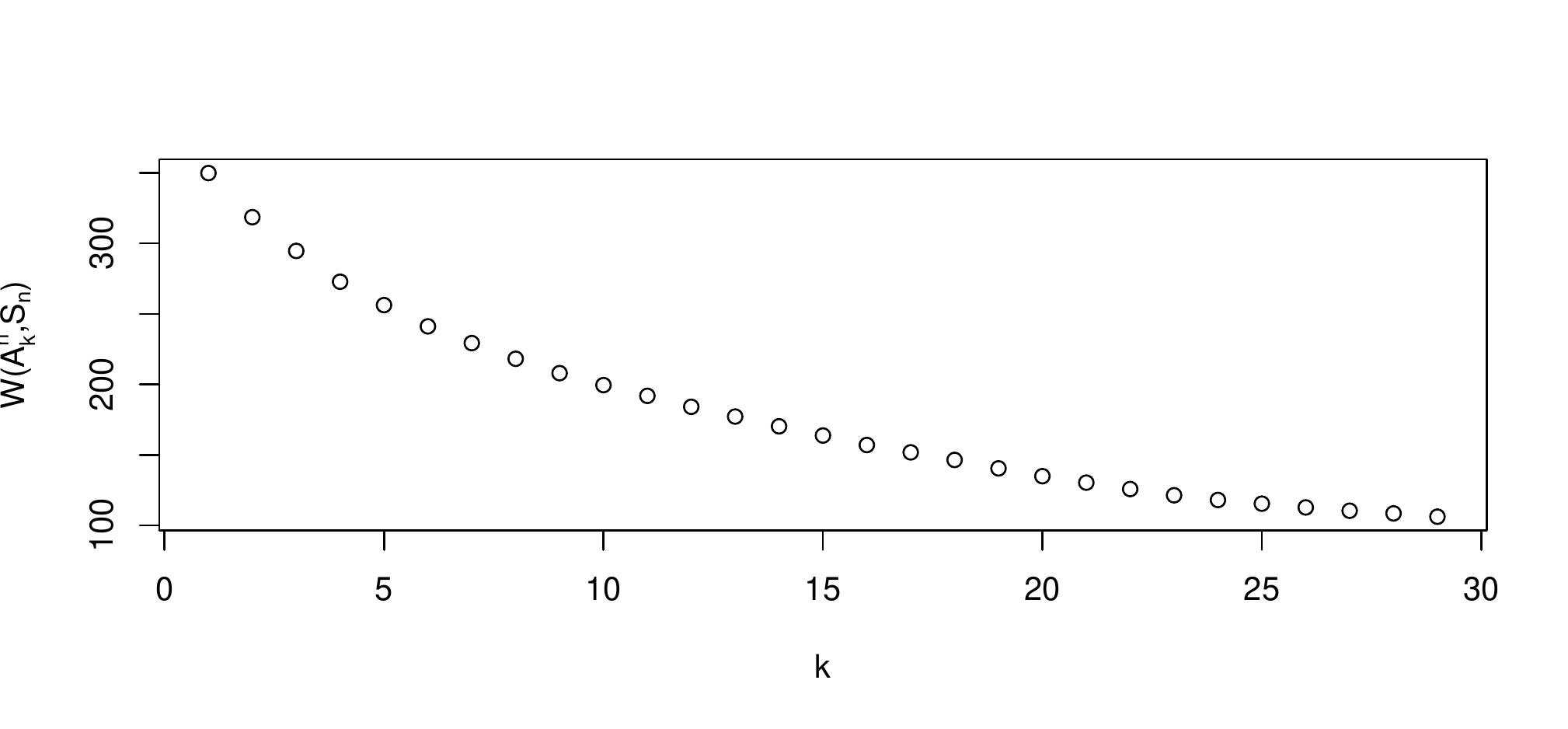}
	\caption{The value of $W(A_k^n,S_n)$ for different values of $k$ in the financial portfolio loss data. }
	\label{Fig:portelbow}
\end{figure}

We create again an ``elbow plot", see Figure~\ref{Fig:portelbow}. Here it is rather difficult to find a concrete value for $k$ so we compare the analysis for $k=5$ and $k=10$. 

\begin{figure}[h]
	\includegraphics[width=\textwidth]{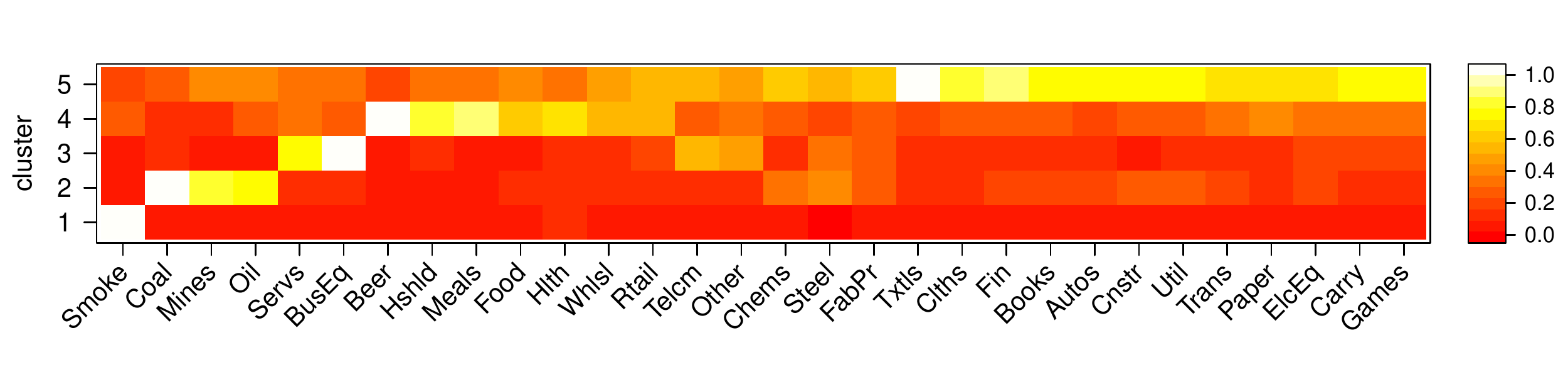}
	\includegraphics[width=.95\textwidth]{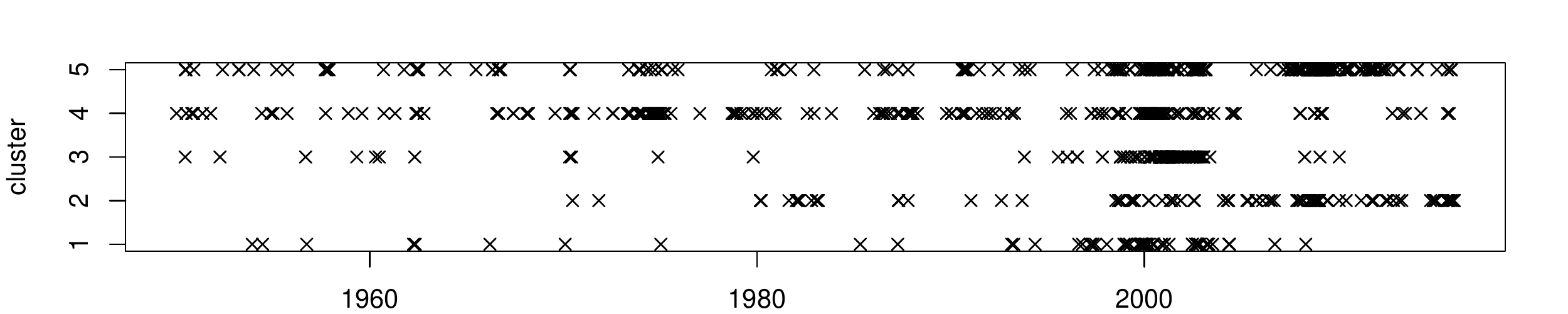}
	\caption{The $k$-means clustering result on the financial portfolio loss data for $k=5$. Top: cluster classification result vs.~time.  Bottom: re-normalized cluster centers illustrated.}
	\label{Fig:port1}
\end{figure}  
\begin{figure}
	\includegraphics[width=\textwidth]{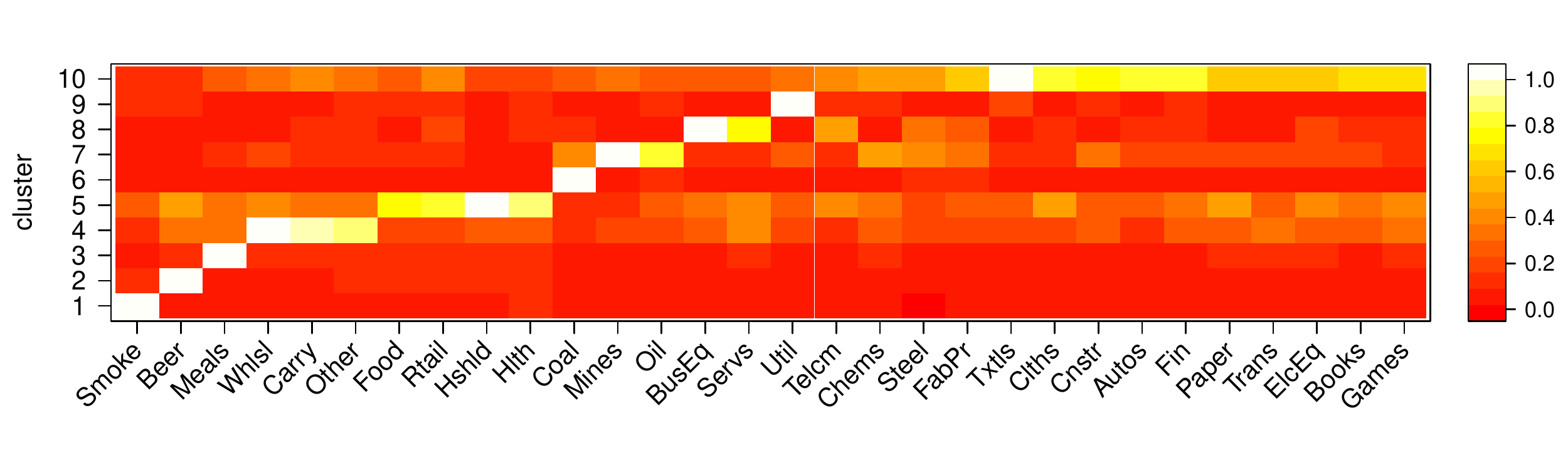}
	\includegraphics[width=.95\textwidth]{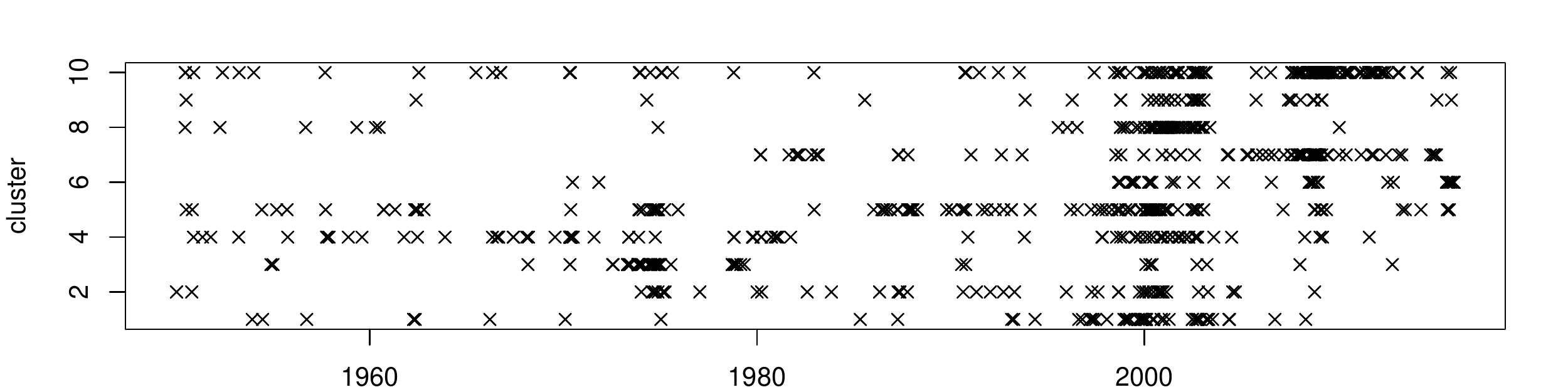}
	\caption{The $k$-means clustering result on the financial portfolio loss data for $k=10$.  Top: cluster classification result vs.~time.  Bottom: re-normalized cluster centers illustrated.}
	\label{Fig:portfolio2}
	\end{figure}

The top graph of Figure~\ref{Fig:port1} illustrates the cluster centers with $k=5$. We can see that the clusters clearly separate the categories into different sectors.  Cluster 1 signals the asymptotic independence of the tobacco industry to all other categories.  Cluster 2 focuses on the energy and material sectors.  Cluster 3 consists of business and IT related industries.  Cluster 4 consists of the consumer oriented industries and Cluster 5 encompasses the rest.  

The top graph of Figure~\ref{Fig:port1}, which shows the time points of the extreme losses in each cluster, also provides interesting insights. The extreme losses for Cluster 1 occurred around 2000, when massive lawsuits surged against tobacco companies. For Cluster 2, since the turn of the millennium, the U.S.\ coal and mining industries have been more heavily affected by government regulations, the rise of alternative sources of energy and foreign imports, which led to many struggles in the industry.  The timeline for Cluster 3 clearly indicates the dot-com bubble in the late nineties. The consumer goods of Cluster 4 were heavily affected by U.S.\ recessions, most prominently by the oil crisis in 1973. Cluster 5 can be explained as widespread effects of the dot-com bubble and financial crisis.

%The choice of $k$ in this example is rather ambiguous and as $k$ increases the number of detected clusters with only one or very few large components increases, hinting at an overall strong level of asymptotic independence. Yet, we can clearly identify sectors which are more prone to exhibiting joint losses, as many of the connections from the analysis with $k=5$ can also be observed when we set $k=10$ as in Figure~\ref{Fig:portfolio2}. Still tightly linked are the consumer sector (Cluster 3), the energy sector (Cluster 7) and a wide collection of industries linked to the financial industry (Cluster 10).  The timeline plot also shows similar patterns to the previous results and clearly identifies the major events in the financial market history.

Figure~\ref{Fig:portfolio2} shows the same set of results for $k=10$.  There is a clear pattern where most clusters have only one or few large components, hinting at an overall strong level of asymptotic independence.  We can clearly identify sectors which are more prone to exhibiting joint losses, and that many of the connections from the analysis with $k=5$ remain. Still tightly linked are the consumer sector (Cluster 5), the energy sector (Cluster 7), the business and IT sectors (Cluster 8) and a wide collection of industries linked to the financial industry (Cluster 10).  The timeline plot also shows similar patterns to the previous results and clearly identifies the major events in the financial market history.

\subsection{Dietary intakes data} \label{Subsec:dietary}

In this section we look at the dietary interview from the 2015--2016 NHANES report, available at \url{https://wwwn.cdc.gov/Nchs/Nhanes/2015-2016/DR1TOT_I.XPT}.  The interview component, called ``What We Eat in America'', recorded the food and beverage consumed by all participants during the 24 hours period prior to the interview.  The resulting dataset describes the nutrients information calculated from these observations.  We are interested in the dependency of 38 chosen nutrients in high-level intakes, as high doses of some of the components can have negative health effects. See also \cite{Chautru15} for the analysis of a similar, but smaller data set.
\begin{figure}[H]
	\centering
	%\includegraphics[height=5cm]{nutrientselbow.pdf}
	%\caption{The value of $W(A_k^n,S_n)$ for different values of $k$ in the dietary intakes data. }
	%\label{Fig:dietelbow}
	\centering
	\includegraphics[width=1.1\textwidth]{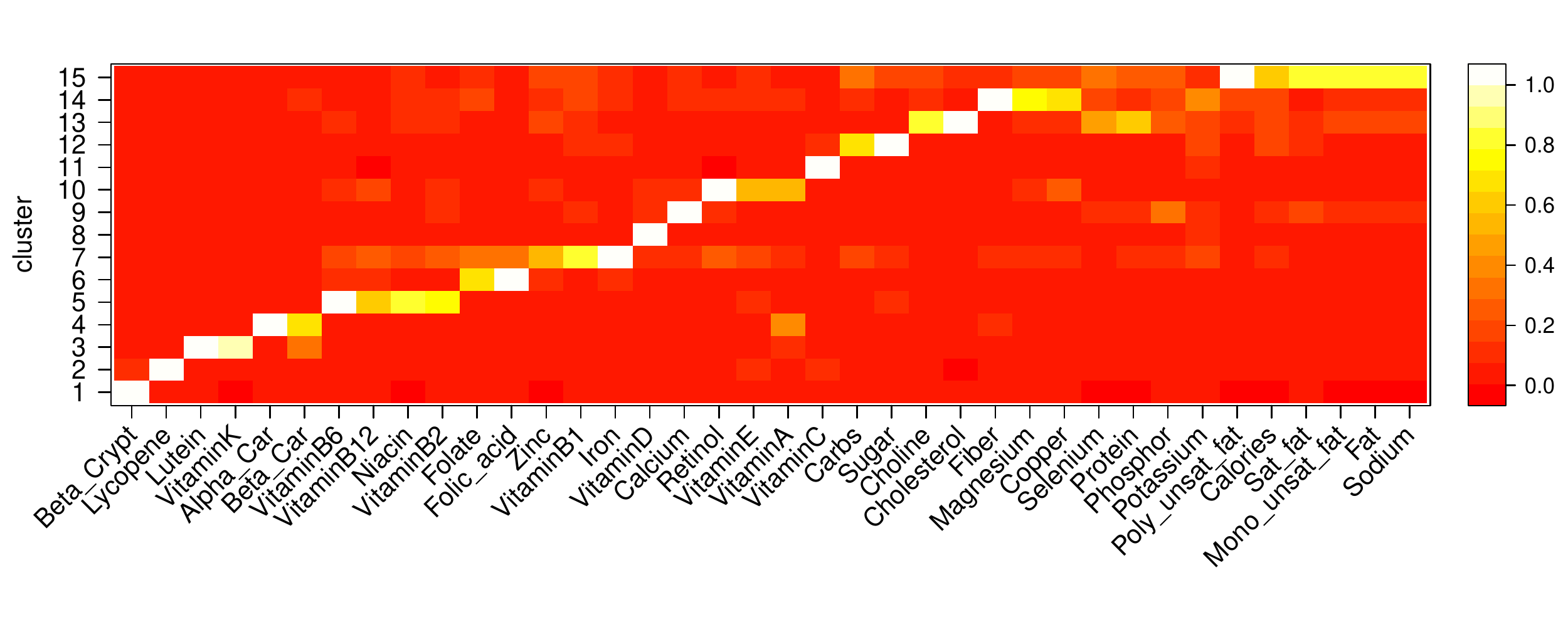}
	\caption{The $k$-means cluster centers in the dietary intakes data for $k=15$. }
	\label{Fig:diet1}
\end{figure}
We derive again an estimator for the spectral measure by transforming observations with the help of the empirical distribution function and keeping the transformed observations whose Euclidean norm belongs to the largest 5\%. The choice of $k$ is again ambiguous in this data set and the elbow plot is similar to that in Figure~\ref{Fig:portelbow}. What is clear is that as $k$ increases the number of cluster centers with only one large component increases, again pointing at asymptotic independence of most of the nutrients. Significant clusters for several values of $k$ can nevertheless be identified as clusters formed by carbs and sugar, by vitamin B$_2$, Vitamin B$_6$, Vitamin B$_{12}$ and niacin, by lutein and vitamin K, by iron and vitamin B$_1$ and finally by fat together with fatty acids which goes also hand in hand with high values of intaken calories. 

\begin{figure}[H]
\centering
\includegraphics[width=1.1\textwidth]{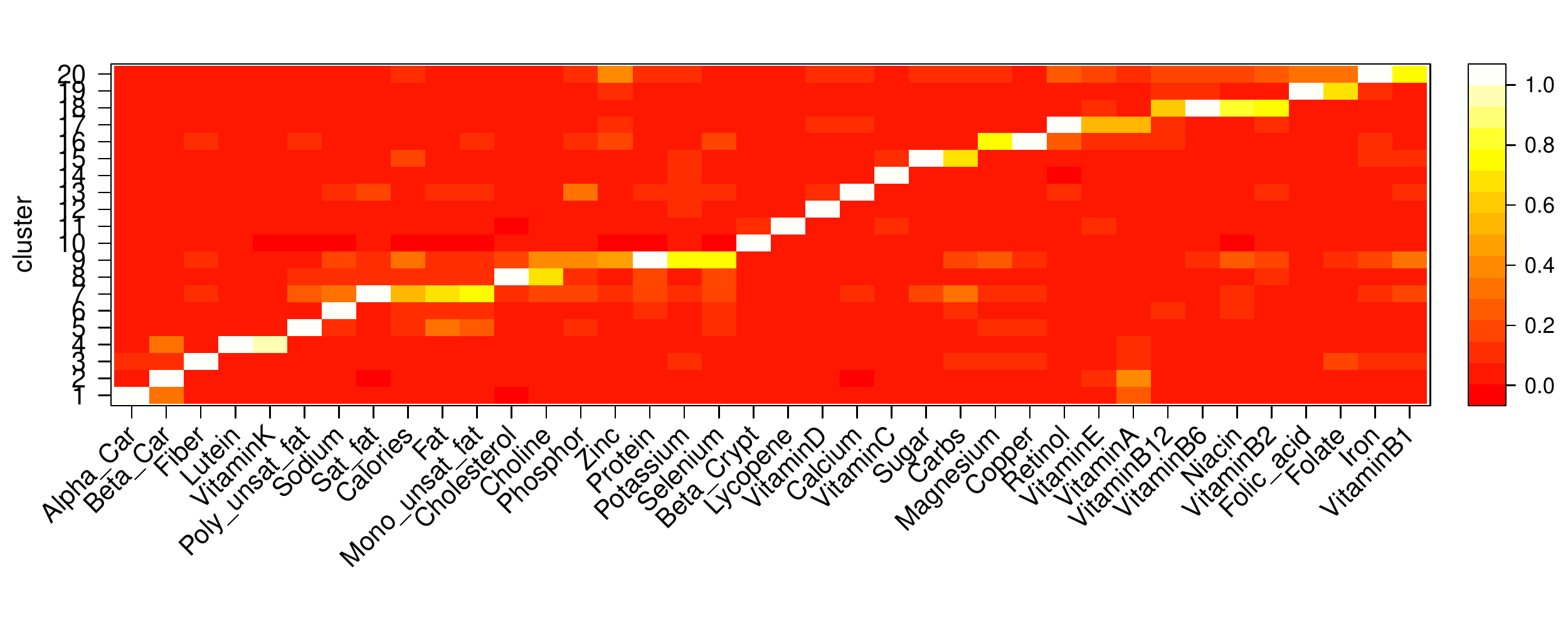}
\caption{The $k$-means cluster centers in the dietary intakes data for $k=20$. }
\label{Fig:diet2}
\end{figure}

\section*{Acknowledgements}
The authors thank Dan Cooley, Holger Drees, Henrik Hult and Chen Zhou for valuable discussions and suggestions regarding this manuscript.
\bibliography{dreductionextremes}
\bibliographystyle{abbrvnat}
\end{document}